\documentclass[12pt,a4paper]{article}

\usepackage{amsmath,amssymb,amsfonts,amsthm}

\usepackage{t1enc}
\usepackage[english]{babel}
\usepackage[latin1]{inputenc}
\usepackage{graphicx}
\usepackage{color}
\usepackage{transparent}

\usepackage{hyperref}

\newtheorem{theorem}{Theorem}[section]

\theoremstyle{definition}

\renewcommand{\b}[1]{\bar{#1}}

\author{Andr\'es E. Ace\~na$^{1}$ and Sergio Dain$^{2,3}$\\
\\
 $^1$ Instituto de Ciencias B\'asicas, ICB, \\
 Universidad Nacional de Cuyo, Mendoza, Argentina\\
 $^2$ Facultad de Matem\'atica, Astronom\'\i a y F\'\i sica, FaMAF, \\
 Universidad Nacional de C\'ordoba \\
 Instituto de F\'\i sica Enrique Gaviola, IFEG, CONICET, \\
 Ciudad Universitaria (5000) C\'ordoba, Argentina\\
 $^3$ Max-Planck-Institut f{\"u}r Gravitationsphysik, Albert Einstein\\
 Institut, Am M\"uhlenberg 1 D-14476 Potsdam Germany}
\date{\today}
\title{Stable isoperimetric surfaces in super-extreme Reissner-Nordstr\"om}

\begin{document}
\maketitle

\begin{abstract}
  We study isoperimetric surfaces in the Reissner-Nordstr\"om spacetime,
  with emphasis on the cuasilocal inequality between area and charge. We
  analyze the stability of the isoperimetric spheres and we found that there
  is a lower bound on the area in terms of the charge, and that the inequality
  is saturated in the transition from the superextremal to the subextremal
  case. We also derive a general inequality between area and charge for stable
  isoperimetric surfaces in maximal electro-vacuum initial data.
\end{abstract}

\section{Introduction}

An important method to obtain physically relevant properties of General
Relativity is through geometrical inequalities. They relate quantities of
physical interest and tell us what type of phenomena is allowed within the
theory. Particularly fruitful have been the search for geometrical inequalities
for axially symmetric black holes (for a recent review see \cite{dain12} and
references therein), where attention to the angular momentum has been paid.
There are two important possible generalizations of these kind of
inequalities. The first one is for non-axially symmetric spacetimes. Axial
symmetry is used in a crucial way to define angular momentum. In order to study
general spacetimes a model problem is to replace angular momentum by
electric charge, this has been done recently in \cite{Dain:2011kb}. The second,
and more difficult, generalization is to consider geometrical inequalities for
general objects (i.e. not only black holes). Remarkably, in this kind of
inequalities the black hole trapped surfaces are replaced by stable
isoperimetric surfaces.  In particular in \cite{Dain:2011kb} the following
cuasilocal geometrical inequality has been obtained,
\begin{equation}
\label{des1}
 A\geq\frac{4}{3}\pi (Q_E^2+Q_M^2),
\end{equation}
where $A$ is the area of a stable isoperimetric surface $\Sigma$ in an
electro-vacuum, maximal initial data, with non-negative cosmological constant
and $Q_E$ and $Q_M$ are the electric and magnetic charges of $\Sigma$. This
inequality tells us that it is not possible to put an arbitrarily large
quantity of charge inside an isoperimetric surface. The requirement of $\Sigma$
being isoperimetrical can not be dropped without further requirements. This can
be seen by looking at the spacetime presented by Bonnor
 \cite{bonnor98}. There, a spacetime is constructed where a spheroidal
distribution of charge is surrounded by electro-vacuum. The solution is such
that the quotient $A/Q^2$ for the surface of the spheroid can be made
arbitrarily small.

Taking into account how \eqref{des1} is obtained it is possible to conjecture
that the inequality is not sharp. To investigate this relation we consider the
Reissner-Nordstr\"om spacetime, which can be considered the simplest
non-trivial electro-vacuum solution of Einstein equations.  We found that in
this case the inequality \eqref{des1} is not sharp, and that the bound of the
area in terms of the charge is obtained in the transition from superextremal to
subextremal. We also isolate the possible cause of \eqref{des1} not being sharp
and present a new sharp inequality.

\section{Main results}

Let us consider a spherically symmetric 3-dimensional metric, written in the form
\begin{equation}\label{met3}
 ds^2=f(r)dr^2+r^2d\Omega^2,\hspace{1cm}d\Omega^2=d\theta^2+\sin^2\theta \, d\phi^2.
\end{equation}
We are interested in the Reissner-Nordstr\"om metric, in which case
\begin{equation}
\label{eq:frn}
 f(r)=\left(1-\frac{2m}{r}+\frac{Q^2}{r^2}\right)^{-1},
\end{equation}
where $m$ is the mass and $Q$ the charge. According to the range of $m$ and $Q$
we have three cases, sub-extreme, $m^2>Q^2$, extreme, $m^2=Q^2$ and
super-extreme, $m^2<Q^2$. In the first two cases, the coordinate $r$ has range
$r_0\leq r\leq\infty$, where $r_0=m+\sqrt{m^2-Q^2}$. In the super-extreme case
the coordinate has range $0\leq r\leq\infty$. 

A surface is called isoperimetric if its area is an extreme with respect to
nearby surfaces that enclose the same volume. This implies that its mean
curvature is constant.  It is also called stable if its area is a minimum. For
further discussion on isoperimetric surfaces in this context we refer to
\cite{Dain:2011kb} \cite{dain12} and for the concept of stability see
\cite{Barbosa88}. We have the following condition on an isoperimetric surface
$\Sigma$ to be stable \cite{Barbosa88},
\begin{equation}
 F(\alpha)>0,
\end{equation}
where
\begin{equation}\label{stabFun}
 F(\alpha)=
\int_\Sigma\left[-\alpha\Delta_\Sigma\alpha-\alpha^2\left(\chi_{AB}\chi^{AB}+R_{ab}n^an^b\right)\right]dA_\Sigma,   
\end{equation}
and $\alpha$ is any function such that
\begin{equation}\label{intalpha}
 \int_{\Sigma}\alpha \, dA_\Sigma=0.
\end{equation}
In \eqref{stabFun} $R_{ab}$ is the three-dimensional Ricci tensor and $n^a$ the
normal, $\chi_{AB}$ the second fundamental form and $dA_\Sigma$ the volume
element of $\Sigma$.

It is known that for $m^2>Q^2$ all
spheres of revolution $r=constant$ are isoperimetric stable surfaces 
(see \cite{corvino07}). We want to analyze the case $m^2<Q^2$.

\begin{theorem}\label{mt1}
Consider the spheres $r=constant$ in the Reissner-Nordstr\"om metric
given by equations \eqref{met3}  \eqref{eq:frn}. Then we have the
following result:
\begin{enumerate}
\item  For $0\leq|Q|\leq m $ all these surfaces are isoperimetric
  stable.

\item For $m<0$ all these surfaces are isoperimetric but not stable. 

\item For $0<m<|Q|$ the surfaces with radius $r>r_c$ are isoperimetric stable. The
  surfaces with $r<r_c$ are unstable, where
  \begin{equation}
    \label{eq:1}
    r_c=\frac{2Q^2}{3m}.
  \end{equation}
In particular, all stable isoperimetric surfaces satisfy the bound
\begin{equation}
  \label{eq:2}
 A \geq \frac{16}{9}\pi Q^2.  
\end{equation} 
\item There is not a sphere in Reissner-Nordstr\"om where the inequality \eqref{eq:2}  is
  saturated. The inequality is saturated in the limit for the sphere $r=r_c$
  when the extreme case is approached from the superextreme case.
\end{enumerate}
\end{theorem}
Remark: we also prove that the stability operator is not positive for test
functions that do not satisfy the volume preserving condition \eqref{intalpha}
(for example the constants).

Note that the bound \eqref{eq:2} is higher than the one obtained in
\cite{Dain:2011kb}. We expect this bound to be optimal. Following the analysis
in \cite{Dain:2011kb} which is based on \cite{christodoulou88} we get the
following result.

\begin{theorem}
\label{t:main3} 
Consider an electro-vacuum, maximal initial data, with a non-negative
cosmological constant. Assume that $\Sigma$ is a  stable isoperimetric
sphere. Then
\begin{equation}
  \label{eq:isoperimetric} 
 \left(1-\frac{1}{16\pi}\chi^2 A \right)  A\geq \frac{4\pi}{3}(Q^{2}_{E}+Q^{2}_{M}), 
\end{equation}
where  $Q_{E}$  and $Q_{M}$ are the electric and magnetic charges of
$\Sigma$ and $\chi$ is its mean curvature.   

Moreover, the  surfaces $r_c$ in super-extreme Reissner-Nordstr\"om satisfy the
equality in (\ref{eq:isoperimetric}).

\end{theorem}

Remark: The inequality \eqref{eq:isoperimetric} in this theorem is a straight
forward consequence of previous results \cite{christodoulou88}. The
interesting and new part of this theorem is that equality is achieved for this
limit surface in Reissner-Nordstr\"om, showing that previously neglected terms
are of consequence.

It is an interesting open problem to study the same problem for super-extreme
Kerr. In that case the problem is much more complex because the location of the
isoperimetric surfaces is known only numerically (see \cite{Metzger:2004pr}). 

Let us discuss the different regimes for the solution and the relation with
respect to the area-charge inequalities. An appropriate quantity to consider for
an isoperimetric stable surface in this context is
\begin{equation}
 \frac{A}{4\pi Q^2}
\end{equation}
as a function of $\epsilon=Q^2/m^2$. The parameter $\epsilon$ is  the
natural parameter for distinguishing the different regimes, where $\epsilon<1$
corresponds to subextremal, $\epsilon=1$ to extremal and $\epsilon>1$ to
superextremal. For the subextreme case we have
\begin{equation}\label{inSub}
 \frac{A_{sub}}{4\pi Q^2}=\frac{r^2}{Q^2}\geq \frac{r_0^2}{Q^2} = -1+\frac{2}{\epsilon}\left(1+\sqrt{1-\epsilon}\right).
\end{equation}
For the superextreme case
\begin{equation}\label{inSuper}
 \frac{A_{super}}{4\pi Q^2}=\frac{r^2}{Q^2}\geq \frac{r_c^2}{Q^2} = \frac{4}{9}\,\epsilon.
\end{equation}
For comparison, the inequality previously obtained in  \cite{Dain:2011kb}  is
\begin{equation}\label{inPr}
 \frac{A}{4\pi Q^2}\geq \frac{1}{3}.
\end{equation}
It is interesting to note that the bounds in \eqref{inSub} and \eqref{inSuper}
appear because there is a limiting inner sphere. In the subextremal case this
is the boundary of the manifold corresponding to the event horizon, while in
the superextreme case it is the transition from stable to unstable surfaces. We
put these inequalities together in the following graph, where the dark gray  region
are spheres in the subextremal case, the light gray region is the superextremal case
and the bottom line is the previously obtained bound. Here it is worth
noticing that the inequality gets close to equality  as one approaches the extreme
case, both from the subextremal and the superextremal cases. Also, there is a
gap between the inequality \eqref{inPr} and the lower bound, suggesting that in
general it is not optimal.

\begin{figure}[h]
  \centering
  \def\svgwidth{.7\textwidth}
  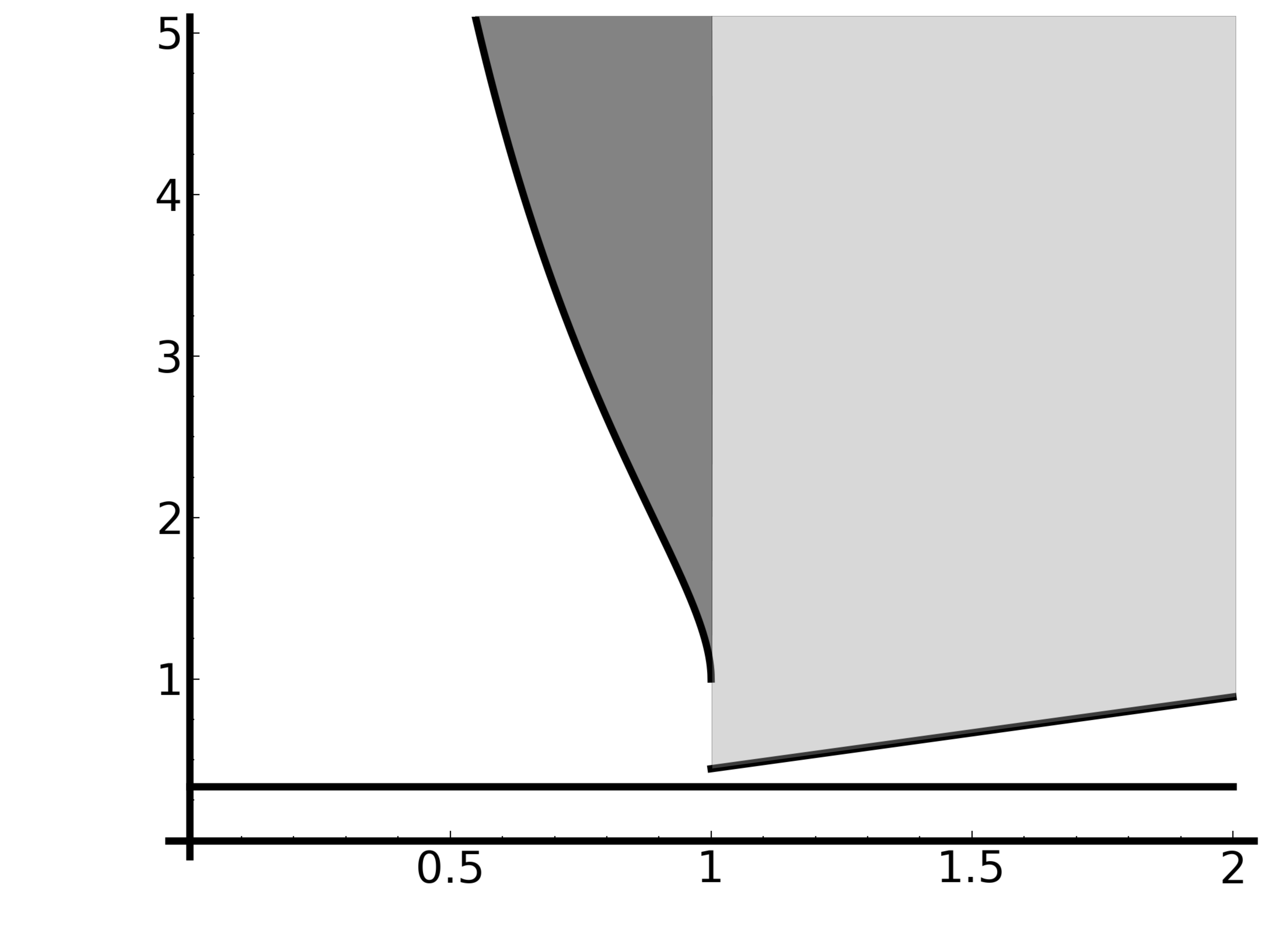
\end{figure}

\section{Proof of the theorems}

\begin{proof}[Proof of theorem \ref{mt1}]
  Let us consider the surface $\Sigma=\{r=constant\}$. From \eqref{met3} and
  \eqref{eq:frn} is a direct calculation to show that the mean curvature is
\begin{equation}
 \chi=\frac{2}{r\sqrt{f}}
\end{equation}
and therefore the surface is isoperimetric. Considering \eqref{stabFun},
\begin{eqnarray}
 F(\alpha) & = &
 \int_{\Sigma}\left[-\alpha\Delta_\Sigma\alpha-\alpha^2\left(\frac{2}{r^2f}+\frac{f'}{rf^2}\right)\right]d\Sigma\\ 
& = & \int_{\Sigma_0}\left[-\alpha\Delta_0\alpha-\alpha^2\frac{1}{f}\left(2+\frac{rf'}{f}\right)\right]d\Sigma_0\\
& = & \int_{\Sigma_0}\left[-\alpha\Delta_0\alpha-\alpha^2 \left(2-\frac{6m}{r}+\frac{4Q^2}{r^2}\right)\right]d\Sigma_0,
\end{eqnarray}
where $\Sigma_0$ is the unit sphere and $\Delta_0$ is the Laplacian on it. The
lowest non-zero eigenvalue $\lambda_1=2$ of the laplacian on the sphere,
$\Delta_0\alpha=-\lambda\alpha$, can be written in the following variational form
\begin{equation}
\label{eq:var} 
\lambda_1=\inf_{\int\alpha \, dS_0=0}\frac{\int|D\alpha|^2dS_0}{\int\alpha^2dS_0}.
\end{equation}
From (\ref{eq:var}) we deduce
\begin{equation}
 \int-\alpha\Delta_0\alpha dS_0=\int|D\alpha|^2dS_0\geq 2\int\alpha^2dS_0,
\end{equation}
and hence
\begin{equation}\label{Falpha}
  F(\alpha) \geq \frac{2}{r}\left(3m-\frac{2Q^2}{r}\right)\int_{S_0}\alpha^2 dS_0
\end{equation}
where we have restricted to functions that satisfy \eqref{intalpha}. In
particular, equality is obtained in \eqref{Falpha} when the function $\alpha$
is an eigenfunction corresponding to $\lambda_1$.  Using this test function
$\alpha$ and the equality in \eqref{Falpha} we see that if $m<0$
then all spheres are unstable. On the other hand, if $|Q|<m$ from the
inequality \eqref{Falpha} we deduce that all spheres are stable, as in this
case $r\geq r_0$. If $0<m<|Q|$ we can define a critical radius
\begin{equation}
 r_c=\frac{2Q^2}{3m},
\end{equation}
such that if $r<r_c$ then the sphere is unstable and if $r>r_c$ it is stable.

The proof of \eqref{eq:2} comes from the analysis of \eqref{inSub} and
\eqref{inSuper}. The minimum of the r.h.s. in the range of applicability of
\eqref{inSub} is attained for the extremal case, that is, $\epsilon=1$, and
gives
\begin{equation}
 \frac{A_{sub}}{4\pi Q^2}\geq 1.
\end{equation}
For \eqref{inSuper} we have also that the minimum is obtained for $\epsilon=1$, although in this case the minimum is obtained as a limit,
\begin{equation}
 \frac{A_{super}}{4\pi Q^2}\geq\frac{4}{9}.
\end{equation}
Comparing the last two inequalities we obtain the bound \eqref{eq:2}. That there is no sphere that actually saturates the inequality is because if we take $\epsilon=1$ then we are in the extremal case and then $r_0>r_c$.
\end{proof}

\begin{proof}[Proof of theorem \ref{t:main3}]
  The proof follows the proof of \eqref{des1} in \cite{Dain:2011kb}. If one
  follows  \cite{christodoulou88}  it is possible to see that for a stable
  isoperimetric surface
\begin{equation}
 12\pi\geq\frac{1}{2}\int_\Sigma\left(R+\frac{3}{2}\chi^2+\b{\chi}_{AB}\b{\chi}^{AB}\right)dA_\Sigma,
\end{equation}
where $\b{\chi}_{AB}$ is the trace-free part of the second fundamental form of
$\Sigma$ and $R$ is the three-dimensional Ricci scalar. The constraint
equations in the three-dimensional manifold imply
\begin{equation}
 R+K^2-K_{ab}K^{ab}-2\Lambda=2(E^2+B^2).
\end{equation}
If we consider now that the data is maximal, $K=0$, and that we assume
$\Lambda\geq0$, combining the previous equations we have
\begin{equation}
 12\pi-\frac{1}{2}\int_\Sigma\left(\frac{3}{2}\chi^2+\b{\chi}_{AB}\b{\chi}^{AB}\right)dA_\Sigma
 \geq \int_\Sigma\left(E^2+B^2\right)dA_\Sigma 
\end{equation}
As shown in  \cite{Dain:2011kb},
\begin{equation}
 \int_\Sigma\left(E^2+B^2\right)dA_\Sigma \geq \frac{16\pi^2}{A}(Q^{2}_{E}+Q^{2}_{M}).
\end{equation}
Using this, neglecting the term $\b{\chi}_{AB}\b{\chi}^{AB}$ as it is always
positive and can be zero, and remembering that $\chi=constant$, we obtain
\eqref{eq:isoperimetric}. It is important to notice that discarding the term
$\b{\chi}_{AB}\b{\chi}^{AB}$ does not pose a risk to the inequality being
sharp, as this term is zero if the surface is umbilical.

If we evaluate \eqref{eq:isoperimetric} for Reissner-Nordstr\"om, we have
\begin{equation}
 4\pi(2mr-Q^2)\geq\frac{4}{3}\pi Q^2,
\end{equation}
which is saturated in the superextremal case for $r=r_c$ and is never saturated
in the subextremal case.
\end{proof}

\section{Acknowledgments}

We would like to thank Jan Metzger for nice discussions. AEA would like to
thank the General Relativity Group of FaMAF for hospitality.

SD is supported by CONICET (Argentina). This work was supported by grants PIP
112-200801-00754 of CONICET (Argentina), Secyt from Universidad Nacional de
C\'ordoba (Argentina), and a Partner Group grant of the Max Planck Institute
for Gravitational Physics (Germany). AEA is supported by SeCTyP from Universidad
Nacional de Cuyo (Argentina) through grant 06/M030.


\begin{thebibliography}{1}

\bibitem{Barbosa88}
J.~L. Barbosa, M.~Do~Carmo, and J.~Eschenburg.
\newblock Stability of hypersurfaces of constant mean curvature in riemannian
  manifolds.
\newblock {\em Mathematische Zeitschrift}, 197(1):123--138, 1988.

\bibitem{bonnor98}
W.~B. Bonnor.
\newblock A model of a spheroidal body.
\newblock {\em Classical and Quantum Gravity}, 15(2):351, 1998.

\bibitem{christodoulou88}
D.~Christodoulou and S.-T. Yau.
\newblock Some remarks on the quasi-local mass.
\newblock In {\em Mathematics and general relativity ({S}anta {C}ruz, {CA},
  1986)}, volume~71 of {\em Contemp. Math.}, pages 9--14. Amer. Math. Soc.,
  Providence, RI, 1988.

\bibitem{corvino07}
J.~Corvino, A.~Gerek, M.~Greenberg, and B.~Krummel.
\newblock On isoperimetric surfaces in general relativity.
\newblock {\em Pacific J. Math.}, 231(1):63--84, 2007.

\bibitem{dain12}
S.~Dain.
\newblock Geometric inequalities for axially symmetric black holes.
\newblock {\em Classical and Quantum Gravity}, 29(7):073001, 2012, 1111.3615.

\bibitem{Dain:2011kb}
S.~Dain, J.~L. Jaramillo, and M.~Reiris.
\newblock {Area-charge inequality for black holes}.
\newblock {\em Class. Quantum Grav.}, 29(3):035013, 2012, 1109.5602.

\bibitem{Metzger:2004pr}
J.~Metzger.
\newblock {Numerical computation of constant mean curvature surfaces using
  finite elements}.
\newblock {\em Class.Quant.Grav.}, 21:4625--4646, 2004, gr-qc/0408059.

\end{thebibliography}

\end{document}